\documentclass[letterpaper, 10 pt, conference]{ieeeconf}

\IEEEoverridecommandlockouts                              
\usepackage[utf8]{inputenc}
\usepackage{include_packages} 
\usepackage[font=footnotesize,skip=5pt]{caption}
\usepackage[english]{babel}

\title{The Cost of Informing Decision-Makers in Multi-Agent Maximum Coverage Problems with Random Resource Values}
\author{{Bryce L. Ferguson, Dario Paccagnan, and Jason R. Marden}
\thanks{This research was supported by ONR Grant \#N00014-20-1-2359, and AFOSR Grants \#FA95550-20-1-0054 and FA9550-21-1-0203}
\thanks{B. L. Ferguson (corresponding author) and J. R. Marden are with the Department of Electrical and Computer Engineering, University of California, Santa Barbara, CA, {\texttt{\{blferguson,jrmarden\}@ece.ucsb.edu}}.}
\thanks{D. Paccagnan is with the Department of Computing, Imperial College London, {\texttt{\{d.paccagnan\}@imperial.ac.uk}}.}
}

\begin{document}
\maketitle

\begin{abstract}
    The emergent behavior of a distributed system is conditioned by the information available to the local decision-makers.
    Therefore, one may expect that providing decision-makers with more information will improve system performance; in this work, we find that this is not necessarily the case.
    In multi-agent maximum coverage problems, we find that even when agents' objectives are aligned with the global welfare, informing agents about the realization of the resource's random values
    can reduce equilibrium performance by a factor of 1/2.
    This affirms an important aspect of designing distributed systems: information need be shared carefully.
    We further this understanding by providing lower and upper bounds on the ratio of system welfare when information is (fully or partially) revealed and when it is not, termed the value-of-informing.
    We then identify a trade-off that emerges when optimizing the performance of the best-case and worst-case equilibrium.
\end{abstract}

\section{Introduction}

In large-scale systems, the prospect of distributing decision-making to local entities is becoming increasingly enticing as a method to reduce complexity while maintaining some level of performance.
This can take the form of swarm control for robotic fleets~\cite{korsah2013comprehensive}, autonomous driving decisions in mobility services~\cite{Wollenstein2020}, local task assignment decisions~\cite{mathew2015planning}, and many more.
Taking a distributed approach entails assigning each agent a decision-making algorithm, such as maximizing an assigned local objective function~\cite{Marden2014}, then analyzing the system's equilibria~\cite{Swenson2018}.
As agents need not possess full knowledge of the overall system; the local decisions (and, ultimately, global behavior) are dependent on the information communicated with and between the agents~\cite{Rahili2017,Ferguson2022robust}
In this work, we address how available information affects system performance.

We focus on maximum coverage problems: a class of models in which each agent selects a set of resources from a ground set, with the objective of maximizing the value of covered resources.
To solve this in a distributed fashion, each agent is given a utility function to evaluate what set of resources to cover; this forms a game played by the agents with their resource selection as their action and the evaluation of their assigned utility function as their payoff.
Existing work has focused on how to design these utility rules and how well the resulting equilibria of the emergent game approximate the optimal welfare~\cite{Paccagnan2018a,Gairing2009,Ramaswamy2022}.
In this work, we generalize this model to consider the case where agents have uncertainty about the resources' values.
In this setting, we ask how revealing information to local decision-makers affects equilibrium performance. Interestingly, we find that \emph{revealing truthful information about the system state can worsen system performance}. While this phenomenon has been observed before in social systems~\cite{Acemoglu2018}, here we find that similar conclusions hold even when the local decision-makers' objectives are \emph{aligned} with the global welfare.

To study this, we consider a \emph{Bayesian persuasion} framework, 
in which a well-informed system operator can strategically reveal information to agents using
messages (or signals) which contain partial information~\cite{Kamenica2011}.
In this work, we study how this information revealing affects equilibrium welfare.
To this end, we introduce a new performance metric termed the \emph{value-of-informing}, which measures the ratio between the equilibrium welfare under an information-revealing policy and when no information is revealed.
This measures the gain or loss in welfare from revealing information.


The framework of Bayesian persuasion has gained traction in the areas of economics, operations research, and engineering, but typically for settings concerned with the behavior and beliefs of human users (e.g., traffic routing apps~\cite{Castiglioni2020,Ferguson2022avoid,Tavafoghi2018}, pricing/investing decisions~\cite{Chen2020signalling}, hybrid work policies~\cite{shahOptimalInformationProvision2022}, etc.).
Results are typically restricted to a binary classification on whether revealing full information helps or not~\cite{sezer2021social,Kamenica2011}, or methods to compute optimal information revealing policies in limited settings~\cite{Zhu2022,wu2019information,Tavafoghi2018}, often with no guarantee on the magnitude of improvement.
Here, we adapt the ideas of information provisioning to the setting of engineered systems, where designed decision-making components can improve their estimate of the system state by receiving relevant messages (e.g., a fleet of surveillance drones receiving live map updates).


Revealing information to local decision-makers can obviously improve the welfare of emergent system equilibria; however, in this work, we find that this is not always the case.
In fact, system performance may degrade by a factor of 1/2 when revealing truthful information to local decision-makers.
The possible loss in system welfare from informing decision-makers comes from two sources (1) the multiplicity of equilibria and (2) the local objective assigned to each decision-maker.
We study the aforementioned value-of-informing when considering the best-case and worst-case equilibria for different local objectives.
Ultimately, we highlight a trade-off between the possible loss from revealing information to best-case and worst-case equilibrium guarantees when agents' local objectives are designed.

\section{Problem Formulation}
\subsection{Maximum Coverage Problem}
Maximum coverage problems have been used to model resource allocation, sensor coverage, job scheduling, and more~\cite{paccagnanUtilityMechanismDesign2022}.
Consider the multi-agent maximum coverage problem, in which $\mc{R} = \{1,\ldots,R\}$ is a finite set of resources.
For each resource $r \in \mc{R}$, let $v_r \geq 0$ be the value of that resource; further, let $v \in \mathbb{R}_{\geq 0}^{|\mc{R}|}$ be the vector containing each resource value.
Let $N = \{1,\ldots,n\}$ be a set of agents, where each agent $i \in N$ can be assigned to cover a subset of resources $a_i \subseteq \mc{R}$.
The set of allowable assignments for each agent is defined by $\actionset_i \subseteq 2^{\mc{R}}$. 
When each agent is assigned, an allocation of agents is denoted $a = (a_1,\ldots,a_n) \in \actionset = \actionset_1 \times \ldots \times \actionset_n$.
Let $(G,v)$ define a maximum coverage problem where $G = (N,\mc{R},\actionset)$.

In an allocation $a$, the system welfare is equal to the total value of resources covered by at least one agent, i.e.,
$W(a;v) = \sum_{r \in \cup_{i \in N}a_i} v_r.$
However, finding an optimal allocation $\Opt{a} \in \argmax_{a \in \actionset} W(a;v)$ is NP-hard~\cite{Gairing2009}.
It is for this reason, we consider a distributed solution technique to approximate this optimal solution.

\subsection{Distributed Decision Making}
Let each agent $i \in N$ possess a local objective function
$$U_i(a;v) = \sum_{r \in a_i} v_r f(|a|_r),$$
which depends on their own action and the actions of every other agent.
This local objective, or \textit{utility function}, is parameterized by a local utility rule $f:\mathbb{N}\rightarrow \mathbb{R}$ that takes as argument $|a|_r$, the number of agents covering resource $r$ in allocation $a$.
The system operator can adopt a utility rule $f$ without exact knowledge of the problem instance if needed.

When agents sequentially and repeatedly update their assignment to one that currently maximizes their utility function, these best-response dynamics have fixed points that are the Nash equilibria of the underlying game.
The convergence of best response dynamics in potential games to pure Nash equilibria are addressed in~\cite{Gairing2009,Swenson2018}.
Nash equilibrium allocations can be defined by
\begin{equation}\label{eq:Nash_def_det}
    U_i(\Nash{a};v) \geq U_i(a_i^\prime,\Nash{a}_{-i};v)~\forall a_i^\prime\in \actionset_i,~i \in N,
\end{equation}
where $a_{-i}$ denotes the allocation of all agents but player $i$.
Let $\NE(G,v,f)$ denote the set of states satisfying \eqref{eq:Nash_def_det}.
These states represent the possible solutions of the distributed dynamics, and we will consider their welfare as an approximation of the maximum coverage problem.

\begin{figure}
    \centering
    \includegraphics[width=0.485\textwidth]{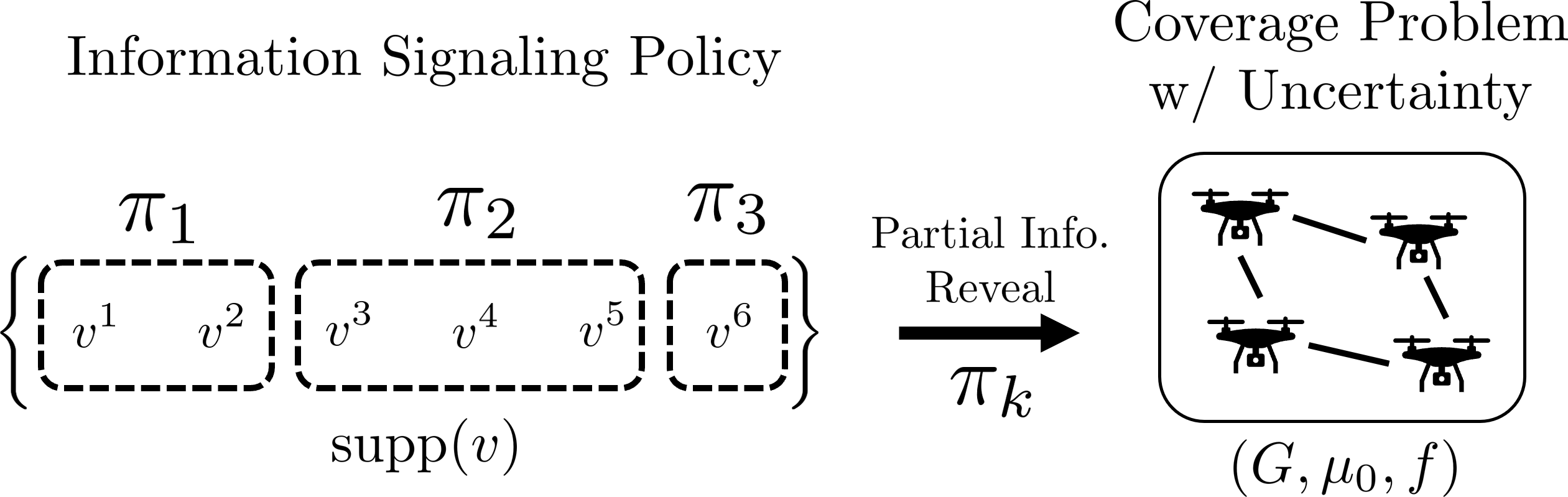}
    \caption{Depiction of information signaling in maximum coverage problems. On the left is the support of a random state variable $v$. 
    At right is a maximum coverage problem, to which we have assigned each agent a local objective.
    The agents in the coverage problem possess the prior distribution of the unknown state variable and receive some partial information $\pi_k$ about the realization.
    The manner in which information is revealed will alter how agents evaluate their objectives and change the emergent behavior.}
    \label{fig:sig_demo}
\end{figure}

\subsection{Uncertainty and Information Signaling}
In this work, we consider how uncertainty and information can affect the efficacy of distributed decision-making.
We consider this uncertainty in the form of randomness for the resources values.
Let $v \in \mathbb{R}_{\geq 0}^{|\mathcal{R}|}$ (the vector containing the value of each resource $r \in \mathcal{R}$) be a discrete random variable with prior distribution $\mu_0$.
A realization of $v$ determines each resource's value $[v_1,\ldots,v_{|\mathcal{R}|}]=v$, i.e., the resource values may be correlated.
Let the support of $v$ be $\statespace := \supp(v)$.

First, we consider the case where agents are \textit{uninformed} about the system state, i.e., they know the prior distribution $\mu_0$ but not the exact realization of $v$.
In this setting, the agents optimize their expected utility, which we will denote
$\ExpU_i(a;\mu_0) = \mathbb{E}_{v \sim \mu_0} \left[ \sum_{r \in a_i} v_r f(|a|_r) \right].$
Let $\NE(G,\mu_0,f) = \NE(G,\mathbb{E}_{v\sim\mu_0}[v],f)$.
Additionally, the objective of the maximum coverage problem is to maximize the expected welfare, i.e.,
$\ExpW(a;\mu_0) = \mathbb{E}_{\mu_0} \left[ \sum_{r \in \cup_{i \in N} a_i} v_r \right].$

As a means to try and improve the performance of the distributed decision-making agents, we may consider revealing information about the realization of the system state.
One option is to reveal full information (or let agents know the realization exactly); however, either due to communication constraints or by design choice, it is often meaningful to reveal only partial information as well.
In line with the broader Bayesian Persuasion framework, consider revealing information with an \textit{information signaling policy} $\Pi = \{\pi_1,\ldots,\pi_m\}$, where $\pi_k \subseteq \statespace$, $\pi_j \cap \pi_k = \emptyset$, and $\bigcup_{k=1}^m \pi_k = \statespace$.
This signaling policy $\Pi$ forms a partition over the support of our random state variable $v$.
The signal $\pi \in \Pi$ is revealed to the agents when $v \in \pi$.\footnote{In this work, we consider signaling policies that are deterministic mappings from state to signal. In general, signal $\pi$ could be drawn randomly based on state $v$. Many of the results easily generalize to this setting, but for ease of exposition and relevance to our problem setting of informing designed decision-makers, we present the results for deterministic signaling by treating $\Pi$ as a partition of $\statespace$.}
The new system operates as follows:
A system operator adopts a signaling policy $\Pi$ and utility rule $f$.
A state $v$ is drawn from $\mu_0$, and all agents are informed of which element of $\Pi$ the realization belongs to.
The corresponding signal $\pi$ is sent, and each player $i \in N$ computes the posterior belief on the realization $\mu_\pi(x) = \mathbb{P}[v=x|v \in \pi] = \mu_0(x)/\left(\sum_{v^\prime\in\pi}\mu_0(v^\prime)\right)$ if $v \in \pi$ and zero otherwise.
The agents then seek to maximize their expected utility with the posterior belief.

Agents may now condition their action on the received signal.
Let $\strat \in \mathcal{A}^\Pi$ denote a joint strategy, where an element $\strat_i(\pi) \in \mathcal{A}$ captures the action agent $i$ takes when they receive signal $\pi$.
In a strategy profile $\strat$, agent $i$ has an expected payoff of
$\ExpU_i(\strat;\mu_0,\Pi) = \mathbb{E}_{v \sim \mu_0}\left[\sum_{r \in \strat_i} v_r f(|\strat|_r) \right]$.
Note that $\strat_i$ is implicitly a function of the received signal $\pi$, which itself is determined by the state variable $v$; as such, each of $\strat_i$, $\pi$, and $v$ are random variables.
The expected welfare becomes
$\ExpW(\strat;\mu_0,\Pi) = \mathbb{E}_{v \sim \mu_0}\left[\sum_{r \in \cup_{i \in N}\strat_i} v_r \right]$.

When agents follow best-response dynamics, the set of fixed points becomes the set of \textit{Bayes-Nash equilibria}, $\BNE(G,\mu_0,\Pi,f)$.
A strategy $\BNash{\alpha}$ in this set satisfies
\begin{equation}\label{eq:Nash_def_bayes}
    \ExpU_i(\BNash{\strat};\mu_0,\Pi) \geq \ExpU_i(\strat_i^\prime,\BNash{\strat}_{-i};\mu_0,\Pi)~\forall \strat_i^\prime \in \actionset_i^\Pi.
\end{equation}
The signaling policy $\Pi$ will alter this equilibria set and thus the guarantees of our distributed solution to the maximum coverage problem.

Our motivation for this work is understanding how giving agents information can affect system welfare.
Due to the multiplicity of equilibria, we consider two perspectives: the \textit{optimistic perspective} in which the system designer cares about the best attainable system performance, and the \textit{pessimistic perspective} in which the system designer cares about the worst possible performance.
For the optimistic perspective, the system designer evaluates a signaling policy $\Pi$ by its effect on the system welfare in the best-case equilibrium; as such, let the optimistic \textit{value-of-informing} with signaling policy $\Pi$ be
\begin{equation*}\label{eq:voip_def}
    \voip(G,\mu_0,\Pi,f) = \frac{\max_{\strat \in \BNE(G,\mu_0,\Pi,f)} W(\strat;\mu_0,\Pi)}{\max_{a \in \NE(G,\mu_0,f)}W(a;\mu_0)},
\end{equation*}
which measures the gain in optimistic welfare by using policy $\Pi$.
Similarly, for the pessimistic perspective, the system operator evaluates a signaling policy $\Pi$ by its effect on the system welfare in the worst-case equilibrium; let the pessimistic value-of-informing be
\begin{equation*}\label{eq:voim_def}
    \voim(G,\mu_0,\Pi,f) = \frac{\min_{\strat \in \BNE(G,\mu_0,\Pi,f)} W(\strat;\mu_0,\Pi)}{\min_{a \in \NE(G,\mu_0,f)}W(a;\mu_0)},
\end{equation*}
which is the same ratio but now with the worst-case equilibrium strategy and allocation.
These values inform a system operator of how revealing information will affect their equilibrium guarantees.
They differ from the well-known price-of-anarchy/stability measures in that they relate two equilibrium performances rather than equilibrium to optimal.

\section{Main Results}\label{sec:main}
The main contribution of this work is in lower and upper bounding the value-of-informing for best- and worst-case equilibria.
These bounds depend on agents' local decision-making process.
In \cref{subsec:mc}, we focus on the case where agents' payoffs are aligned with the system objective and find that revealing information improves the best-case equilibrium ($\voip \geq 1$) but can worsen the worst-case equilibrium ($\voim \leq 1$).
In \cref{subsec:gen}, we generalize these bounds to any local utility rule $f$.
In \cref{subsec:trade}, we observe a trade-off in the lower bounds on $\voip$ and $\voim$; \cref{fig:VoI_tradeoff} characterizes this trade-off and highlights the fact that altering the local objectives of agents affects the efficacy of revealing information.

\subsection{Marginal Contribution}\label{subsec:mc}

The first utility design we will consider is the marginal contribution, where each agent makes decisions that maximize their contribution to the system welfare, i.e.,
$U_i(a) = W(a) - W(a_{-i}).$
This can be expressed by the local utility rule
$f^{\rm mc}(x) := \indicator{x=1}.$
When agents follow this utility rule, their preferences are aligned with global welfare.
This utility rule has the property that it maximizes the best-case equilibrium guarantee, known as the price-of-stability ratio~\cite{Ramaswamy2022}.
However, as of yet, the effect of revealing information to decision-makers using this utility function has not been addressed.
In \cref{thm:mc}, we address this question by providing lower and upper bounds on the value-of-informing for the best- and worst-case equilibria.

\begin{theorem}\label{thm:mc}
In a Bayesian Maximum Coverage problem $(G,\mu_0,\Pi)$,  with utility rule $f^{\rm mc}$, the value-of-informing for the best-case equilibrium satisfies
\begin{subequations}
\begin{equation}\label{eq:thm1_voip}
    1 \leq \voip(G,\mu_0,\Pi,f^{\rm mc}) \leq \lvert\Pi\rvert,
\end{equation}
and for the worst-case equilibrium satisfies
\begin{equation}\label{eq:thm1_voim}
    1/2 \leq \voim(G,\mu_0,\Pi,f^{\rm mc}) \leq 2\lvert\Pi\rvert.
\end{equation}
\end{subequations}
All of these bounds are tight, but the upper bounds on $\voim$. 
\end{theorem}
Before proving the statement, we discuss the consequences of these results.
We first see that revealing more information (increasing the cardinality of $\Pi$) provides significant opportunities for improvement in either perspective\footnote{A more refined understanding of this improvement can be attained by considering the difference between the realizations and the prior mean.}.
However, revealing information need not always have such a positive effect.
In the optimistic perspective, revealing information can only improve performance ($\voip \geq 1$); however, doing so does not come without consequence, as revealing information can reduce the quality of the worst-case equilibrium by a factor of 1/2 ($\voim = 1/2$).
This fact affirms an important property of information in a multi-agent system: revealing information must be done carefully.
The proof relies on the following lemma characterizing Bayes-Nash equilibria and the expected welfare.
\begin{lemma}\label{lem:BNE_average}
A joint strategy $\alpha$ is a Bayes-Nash equilibrium if and only if $(\alpha_1(\pi),\ldots,\alpha_n(\pi)) \in \NE(G,\mathbb{E}[v\mid \pi_k],f)$ for each $\pi \in \Pi$.
Additionally, the expected welfare of a joint strategy $\alpha$ is a weighted average of the welfare of the joint actions $\alpha(\pi)$ in the respective deterministic games, i.e.,
$$\ExpW(\strat;\mu_0,\Pi) = \sum_{k=1}^m p_kW\left(\strat(\pi_k);\mathbb{E}[v|\pi_k]\right),$$
where $p_k = \sum_{v \in \pi_k}\mu_0(v)$.
\end{lemma}
\begin{proof}
Let $\alpha \in \mathcal{A}^\Pi$ denote a joint strategy.
We show the first claim by observing the following transformation for any $\alpha$,
\begin{subequations}\label{eq:Util_expec_proof}
\begin{align}
    \ExpU_i(\strat;\mu_0,\Pi) &= \mathbb{E}\left[ \sum_{r \in \mathcal{R}} \mathbb{E}[v_r|\pi]f(|\alpha(\pi)|_r)\right]\label{eq:Util_expec_proof_b}\\
    &= \sum_{k=1}^m p_k U_i\left(\alpha(\pi_k);\mathbb{E}[v |\pi_k]\right),\label{eq:Util_expec_proof_c}
\end{align}
\end{subequations}
where \eqref{eq:Util_expec_proof_b} holds from the law of total expectation.
Because $\alpha$ can be any $m$-tuple of joint actions, \eqref{eq:Nash_def_bayes} is satisfied if and only if
$U_i\left(\alpha(\pi);\mathbb{E}[v |\pi]\right) \geq U_i\left(a_i^\prime,\alpha_{-i}(\pi);\mathbb{E}[v |\pi]\right),$
for all $a^\prime_i \in \mathcal{A}_i$, $\pi \in \Pi$; or, that $\alpha(\pi)$ is a Nash equilibrium for the deterministic game $G$ with values $\mathbb{E}[v |\pi]$ for each $\pi \in \Pi$.
The second claim follows \eqref{eq:Util_expec_proof} with welfare in place of utility.
\end{proof}

\noindent\textit{Proof of \cref{thm:mc}:}
\textit{Best-case equilibrium} - 
We will make use of the function $W^\star(v)$, which denotes the welfare of an optimal allocation in $G$ when the values are $v$.
We note that with the marginal contribution utility rule, each optimal allocation $a^{\rm opt}$ is an equilibrium; thus, the welfare of the best Nash equilibrium is the optimal welfare~\cite{Ramaswamy2022}, or $W^\star(v) = \max_{a \in \mathcal{A}}W(a;v) = \max_{\Nash{a} \in \NE(G,v,f^{\rm mc})}W(\Nash{a})$.
We first make several observations about the function $W^\star$.
Observe that
$W^\star(v) = \max_{a \in \mathcal{A}} \sum_{r \in \mathcal{R}} v_r \indicator{|a|_r > 0},$
is the point-wise maximum of a set of affine (and thus convex) functions of $v$, which is itself convex.
Further, $W^\star$ is positively homogeneous, i.e., $W^\star(\lambda v) = \lambda W^\star(v)$ for all $\lambda \geq 0$ and $v \geq 0$, and $W^\star$ is monotone in $v$, i.e., $v \succeq v^\prime \Rightarrow W^\star(v) \geq W^\star(v^\prime)$ where ``$\succeq$" denotes the element-wise inequality.

Using the properties of $W^\star$, we will prove the bounds on $\voip$.
First, the lower bound.
Consider the Bayesian covering game $(G,\mu_0,\Pi)$.
Observe that
\begin{subequations}\label{eq:w_star_jensen}
    \begin{align}
    \max_{a \in \NE(G,\mu_0,f^{\rm mc})} &\ExpW(a;\mu_0) 
    = W^\star\left(\sum_{k =1}^m p_k \mathbb{E}[v|\pi_k]\right)\label{eq:w_star_jensen_a}\\
    &\leq \sum_{k =1}^m p_kW^\star(\mathbb{E}[v|\pi_k])\label{eq:w_star_jensen_b}\\
    &= \max_{\strat \in \BNE(G,\mu_0,\Pi,f^{\rm mc})} \ExpW(\strat;\mu_0,\Pi),\label{eq:w_star_jensen_c}
\end{align}
\end{subequations}
where \eqref{eq:w_star_jensen_a} holds from the fact the maximum-welfare Nash equilibrium is a system optimum and the second claim of \cref{lem:BNE_average},
\eqref{eq:w_star_jensen_b} holds from $W^\star$ convex and Jensen's inequality, and \eqref{eq:w_star_jensen_c} holds from the second claim of \cref{lem:BNE_average} and the first claim of \cref{lem:BNE_average} with the fact the maximum-Nash is a system optimum.
Rearranging terms gives the first inequality in \eqref{eq:thm1_voip}.
It is tight when $|\mathcal{V}|=1$.

Now, we consider the upper bound on $\voip$.
\begin{subequations}\label{eq:voip_mc_up}
    \begin{align}
    &\max_{\alpha \in \BNE(G,\mu_0,\Pi,f^{\rm mc})} \ExpW(\alpha;\mu_0,\Pi) 
    = \sum_{k =1}^m p_kW^\star(\mathbb{E}[v|\pi_k])\label{eq:voip_mc_up_a}\\
    &\hspace{20pt}= \sum_{k=1}^m W^\star\left(p_k\mathbb{E}[v|\pi_k]\right)\leq \sum_{k=1}^m W^\star(\mathbb{E}_{\mu_0}[v])\label{eq:voip_mc_up_b}\\
    &\hspace{40pt}= \lvert \Pi \rvert \left(\max_{a \in \NE(G,\mu_0,f^{\rm mc})} \ExpW(a;\mu_0)\right),\label{eq:voip_mc_up_d}
\end{align}
\end{subequations}
where \eqref{eq:voip_mc_up_a} holds from the fact a Bayes-Nash joint strategy $\alpha$ is an $m$-tuple of Nash equilibria, the maximum-welfare Nash equilibrium is optimal in $W(\cdot~;v)$ and the second claim in \cref{lem:BNE_average}.
\eqref{eq:voip_mc_up_b} holds from $W^\star$ positive homogeneous and the monotonicity of $W^\star$; more specifically, 
\begin{multline*}
    \mathbb{E}_{\mu_0}[v_r] = \sum_{k \in [m]}p_k\mathbb{E}[v_r|\pi_k]\\ = p_k\mathbb{E}[v_r|\pi_k] + \sum_{k^\prime \in [m]\setminus k}p_{k^\prime}\mathbb{E}[v_r|\pi_{k^\prime}] \geq p_k\mathbb{E}[v_r|\pi_k],
\end{multline*}
holds $\forall$ $r \in \mathcal{R}$, $\pi_k \in \Pi$.
\eqref{eq:voip_mc_up_d} holds from the definition $W^\star$.

To see this is tight, consider a problem with $R$ resources and one agent who can select a single one of them, i.e., $\mathcal{A}_1 = \{1,\ldots,R\}$.
Each resource can take on one of two values: $0$ or $1$.
The prior $\mu_0$ is that exactly one resource is ever of value 1 with equal probability; so the support of $v$ has $R$ elements, each of which occurs with probability $1/R$.
When uninformed, the single agent is indifferent over their actions and cannot attain a payoff greater than $1/R$.
When informed, the single agent can always select the resource of value 1.
Thus ${\rm VoI}^+ = R = |\statespace| = |\Pi|$.

\noindent\textit{Worst-case equilibrium} - 
To focus on equilibrium strategies, let $\Nash{a}(v)$ denote a Nash equilibrium joint action in the game $G$ when the values are $v$.
The following steps will hold for any Nash equilibrium, and the proof will be completed by considering $\Nash{a}(v)$ as the welfare minimizing Nash equilibrium.
Observe that the uninformed Nash welfare satisfies
\begin{multline*}
    W(\Nash{a}(\mathbb{E}_{\mu_0}[v]);\mathbb{E}_{\mu_0}[v]) \leq W^\star(\mathbb{E}_{\mu_0}[v])\leq \\
    \sum_{k=1}^m p_k W^\star(\mathbb{E}[v|\pi_k])\hspace{-1.5pt} \leq  \hspace{-1.5pt}\sum_{k=1}^m p_k 2W(\Nash{a}(\mathbb{E}[v|\pi_k]);\mathbb{E}[v|\pi_k]),
\end{multline*}
where the first inequality holds from the definition of $W^\star$, the second holds from properties of $W^\star$ shown in the first part of the proof, and the third holds from the price-of-anarchy bound of 1/2~\cite{Vetta2002}.
Letting $\Nash{a}(v)$ be the worst-case Nash equilibrium when the values are $v$ in each occurrence, the rightmost expression is the worst-case welfare in a Bayes-Nash joint strategy via \cref{lem:BNE_average}.
This gives the first inequality in \eqref{eq:thm1_voip}.

To see that this bound is tight, consider a resource allocation game with three resources, $\mathcal{R} = \{1,2,3\}$, and two players, $N = \{1,2\}$ with two actions each: $\mathcal{A}_1 = \{r_1,r_2\}$ and $\mathcal{A}_2 = \{r_2,r_3\}$.
Let resource $r_1$ have value $v_1=1$ and $r_3$ have value $v_3=0$.
Let the value of resource $r_2$ be a random variable, where $v_2 = 1-\varepsilon$ with probability $1-p$ and $v_2 = 1+\varepsilon(1-p)$ with probability $p$.
When the agents are not informed of the realization of $v_2$, its expected value is $\mathbb{E}[v_2] = 1 - \varepsilon(1-p)^2< 1$ and there is a unique Nash equilibrium of $a = (r_1,r_2)$, providing a welfare of $W(a;\mu_0) = 2- \varepsilon(1-p)^2$.
When the full revelation signaling policy $\Pi = \{\{r_1\},\{r_2\}\}$ is used, then when $v_2 = 1-\varepsilon$, there is a unique equilibrium of $\strat^1 = (r_1,r_2)$, and when $v_2 = 1+\varepsilon(1-p)$ there are two equilibria, the worse of which is $\strat^2 = (r_2,r_3)$.
This gives an expected welfare of $W(\strat;\mu_0) = (1-p)(2-\varepsilon) + p(1+\varepsilon(1-p))$ and a gain of informing agents of 
${\rm VoI}^-(\Pi,f^{\rm mc}) = \frac{(1-p)(2-\varepsilon) + p(1+\varepsilon(1-p))}{2-\varepsilon(1-p)^2}.$
Letting $p \rightarrow 1$ and $\varepsilon \rightarrow 0$ we get ${\rm VoI}^-(\Pi,f^{\rm mc}) \rightarrow 1/2$.

Finally, we show the upper bound on $\voim$. Observe that the welfare of a Bayes-Nash strategy $\alpha = \{\Nash{a}(\mathbb{E}[v|\pi])\}_{\pi \in \Pi}$ under the signaling policy $\Pi$ satisfies
\begin{multline*}
    \sum_{k=1}^m p_k W(\Nash{a}(\mathbb{E}[v|\pi_k]);\mathbb{E}[v|\pi_k]) \leq \sum_{k=1}^m p_k W^\star(\mathbb{E}[v|\pi_k])\\
    \leq |\Pi|\cdot W^\star(\mathbb{E}_{\mu_0}[v]) \leq 2|\Pi|\cdot W(\Nash{a}(\mathbb{E}_{\mu_0}[v]);\mathbb{E}_{\mu_0}[v]),
\end{multline*}
where the first term is the expected payoff of $\alpha$, and the first inequality holds from $W^\star$ being the optimal welfare, the second holds from the upper bound on $\voip$, and the third holds from the price-of-anarchy bound.
\hfill$\qed$

\begin{figure}
    \centering
    \includegraphics[width=0.4\textwidth]{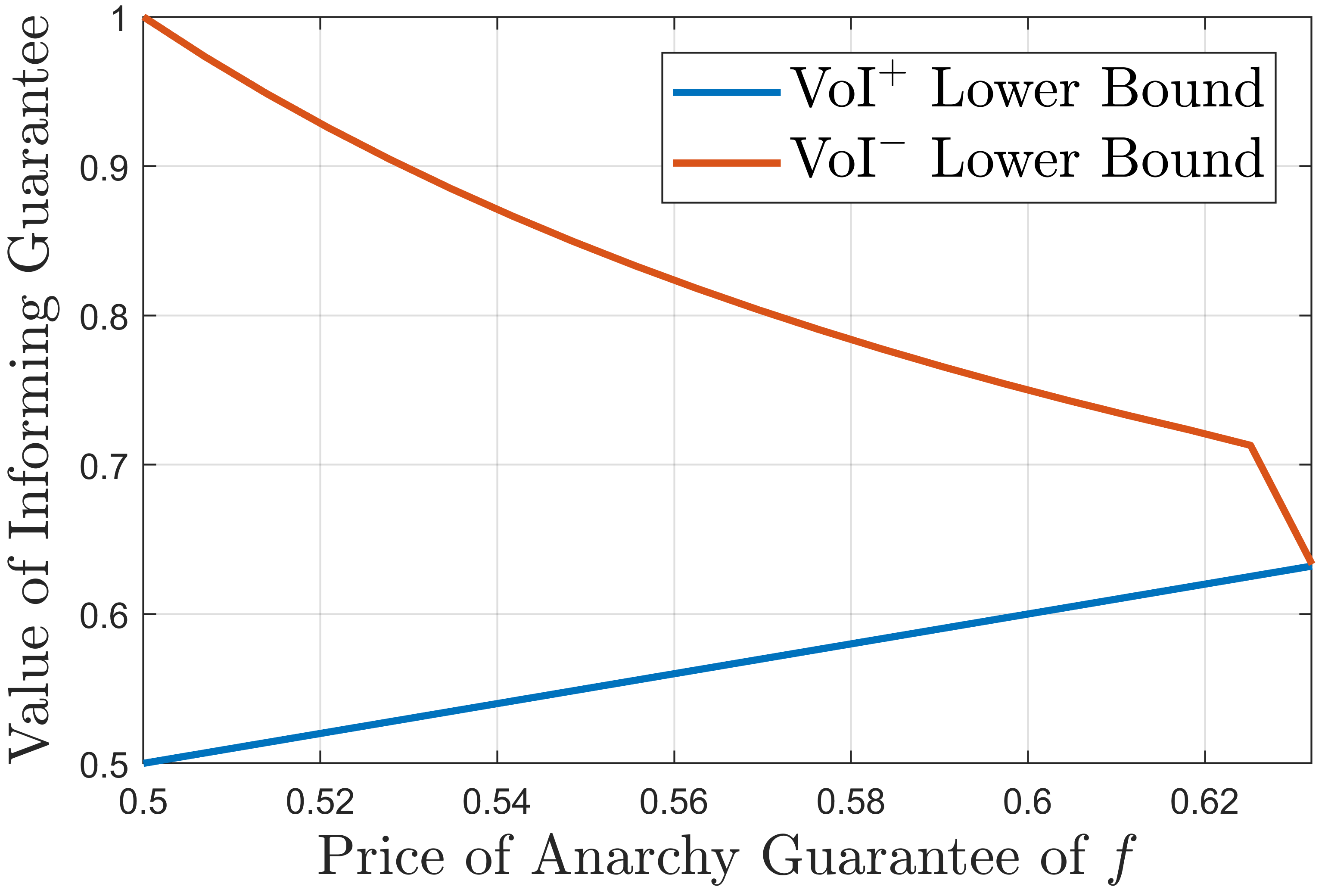}
    \caption{Lower bounds on the value of revealing information for the best- and worst-case equilibria.
    If the utility rule $f$ is designed to lessen the loss to worst-case equilibria (increasing $\voim$), then there is a greater possible loss to the worst-case equilibrium (decreasing $\voip$).
    This trade-off matches the tight bounds from \cref{thm:mc} and \cref{prop:gair}, which appear as the endpoints of this plot.
    The bounds are generated by comparing the value-of-informing to the price-of-anarchy/stability via \cref{thm:gen} and the price of stability to the price-of-anarchy via \cite[Theorem~4.1]{Ramaswamy2022}.
    }
    \label{fig:VoI_tradeoff}
\end{figure}

\subsection{Utility Design \& Informing Efficacy}\label{subsec:gen}

In \cref{subsec:mc}, we considered the special case in which agents' local objectives were aligned with the global objective using the marginal cost utility rule $f^{\rm mc}$.
In this section, we generalize this result to any utility rule by leveraging a connection to the \textit{price-of-stability} and \textit{price-of-anarchy}.

In the deterministic setting, the price-of-stability/anarchy are used to quantify how the best-case and worst-case equilibria approximate the system optimal.
These metrics can be generalized to the Bayesian setting but need not be informative or insightful on the effects revealing information within a single problem instance~\cite{Bhaskar2015} (i.e., failing to capture the benefits and consequences or comparing bounds derived from different problem instances).
Let $\poa(G,v,f) = \frac{\min_{\Nash{a} \in \NE(G,v,f)}W(\Nash{a};v)}{\max_{\Opt{a} \in \mathcal{A}}W(\Opt{a};v)}$ denote the price-of-anarchy for a deterministic maximum coverage problem, and let $\pos(G,v,f) = \frac{\max_{\Nash{a} \in \NE(G,v,f)}W(\Nash{a};v)}{\max_{\Opt{a} \in \mathcal{A}}W(\Opt{a};v)}$ denote the price-of-stability.

Though not immediately apparent, we establish a connection between the price-of-stability/anarchy in deterministic covering games and the value-of-informing in Bayesian covering games.
In \cref{thm:gen}, we leverage this connection to generate bounds on $\voip$ and $\voim$ for any utility design.

\begin{theorem}\label{thm:gen}
Let $\psi := \inf_{v \in \conv{\statespace}} \pos(G,v,f)$ and $\rho := \inf_{v \in \conv{\statespace}} \poa(G,v,f)$, then the value of informing for the best-case equilibrium satisfies
\begin{subequations}
\begin{equation}
    \psi \leq \voip(G,\mu_0,\Pi,f) \leq \psi^{-1}|\Pi|,
\end{equation}
and for the worst-case equilibrium satisfies
\begin{equation}
    \rho \leq \voim(G,\mu_0,\Pi,f) \leq \rho^{-1}|\Pi|.
\end{equation}
\end{subequations}
\end{theorem}

\noindent\textit{Proof of \cref{thm:gen}}:
The proof will rely on the function $W^\star(v) = \max_{a \in \mathcal{A}} W(a;v)$ and several of its properties shown in the proof of \cref{thm:mc}.
First, we prove the bounds on $\voip$.
Let $a \in \NE(G,\mathbb{E}_{\mu_0}[v],f)$ be an arbitrary Nash equilibrium in the deterministic game $G$ with values $\mathbb{E}_{\mu_0}[v]$ and utility rule $f$, and let $\alpha \in \BNE(G,\mu_0,\Pi,f)$ be an arbitrary Bayes-Nash equilibrium in $G$ with prior $\mu_0$ on $v$ and information signaling policy $\Pi$ and utility function $f$.
Observe that the expected welfare of $\alpha$ satisfies
\begin{subequations}\label{eq:gen_ub}
\begin{align}
\ExpW(\alpha;\mu_0,\Pi) &= \sum_{k=1}^m p_k W(\alpha(\pi_k);\mathbb{E}[v|\pi_k])\label{eq:gen_ub_a}\\
&\leq \sum_{k=1}^m p_k W^\star(\mathbb{E}[v|\pi_k])\label{eq:gen_ub_b}\\
&\leq |\Pi|W^\star(\mathbb{E}_{\mu_0}[v])\label{eq:gen_ub_c}\\
&\leq \rho^{-1}|\Pi|\ExpW(a;\mathbb{E}_{\mu_0}[v]),\label{eq:gen_ub_d}
\end{align} 
\end{subequations}
where \eqref{eq:gen_ub_a} holds from \cref{lem:BNE_average}, \eqref{eq:gen_ub_b} holds from the definition of $W^\star$, \eqref{eq:gen_ub_c} holds from the monotonicity and positive homogeneity of $W^\star$ (previously shown in \eqref{eq:voip_mc_up_b}-\eqref{eq:voip_mc_up_d}), and \eqref{eq:gen_ub_d} holds from the definition of $\rho$.

Similarly, we can show that the expected total welfare of the uninformed equilibrium $a$ satisfies
\begin{subequations}\label{eq:gen_lb}
\begin{align}
\ExpW(a;\mu_0) &\leq W^\star(v)\label{eq:gen_lb_a}\\
&\leq \sum_{k=1}^m p_k W^\star(\mathbb{E}[v|\pi_k])\label{eq:gen_lb_b}\\
&\leq \sum_{k=1}^m p_k \rho^{-1} W(\alpha(\pi);\mathbb{E}[v|\pi_k]) \label{eq:gen_lb_c}\\
&= \rho^{-1}W(\alpha;\mu_0,\Pi),\label{eq:gen_lb_d}
\end{align} 
\end{subequations}
where \eqref{eq:gen_lb_b} holds from the convexity of $W^\star$, \eqref{eq:gen_lb_c} holds from the definition of $\rho$, and \eqref{eq:gen_lb_d} holds from \cref{lem:BNE_average}.
Because \eqref{eq:gen_ub} and \eqref{eq:gen_lb} hold for any $a \in \NE(G,\mathbb{E}_{\mu_0}[v],f)$ and $\alpha \in \BNE(G,\mu_0,\Pi,f)$, it holds for each being the respective welfare minimizing equilibria.
If we consider only the case where $a$ and $\alpha$ are the welfare maximizing equilibria, i.e., $a \in \argmax_{a^\prime \in \NE(G,\mathbb{E}_{\mu_0}[v],f)}W(a^\prime;\mathbb{E}_{\mu_0}[v])$ and $\alpha \in \argmax_{\alpha^\prime \in \BNE(G,\mu_0,v,f)}\ExpW(\alpha^\prime;\mu_0,\Pi)$
\hfill\qed

\subsection{Optimistic\hspace{2pt}/\hspace{2pt}Pessimistic Trade-Off}\label{subsec:trade}
\cref{thm:gen} highlighted the fact that altering the utility design will change the impact of information revealing; however, the given bounds need not be tight.
It appears that using a utility design with a higher price-of-anarchy in the deterministic setting will lead to an improved lower bound on $\voim$.
As such, we will more closely consider the price-of-anarchy maximizing rule $$f^{\rm g}(x) := (x-1)!\frac{\frac{1}{(n-1)(n-1)!}+\sum_{i=x}^{n-1}\frac{1}{i!}}{\frac{1}{(n-1)(n-1)!}+\sum_{i=1}^{n-1}\frac{1}{i!}},$$
when $x>0$, and $f^{\rm g}(0)=0$, proven optimal in~\cite{Gairing2009}.
In \cref{prop:gair}, we show tight lower bounds on the value-of-informing while using $f^g$.
Interestingly, though the pessimistic guarantee has improved, the optimistic guarantee has worsened.

\begin{proposition}\label{prop:gair}
While using the the price-of-anarchy maximizing rule $f^{\rm g}$, the value-of-informing for the best-case equilibrium satisfies
\begin{subequations}
\begin{equation}\label{eq:prop_voip}
    1-\frac{1}{e} \leq \voip(G,\mu_0,\Pi,f^{\rm g}),
\end{equation}
and for the worst-case equilibrium satisfies
\begin{equation}
    1-\frac{1}{e} \leq \voim(G,\mu_0,\Pi,f^{\rm g}),
\end{equation}\label{eq:prop_voim}
\end{subequations}
Further, each of these lower bounds is tight. 
\end{proposition}
\noindent\textit{Proof of \cref{prop:gair}}:
\cref{thm:gen} can be used to show that each of \eqref{eq:prop_voip} and \eqref{eq:prop_voim} are valid lower bounds.
To see these bounds are tight, consider a maximum coverage problem with resource set $\mathcal{R} = \{r_{p_i}\}_{i=1}^{n-1} \bigcup \{r_{s_j}\}_{j=1}^z$ where $z = \lceil 1/(f^g(n)-\varepsilon) \rceil$.
The first $n-1$ players can select a single resource from the public resources $\{r_{s_j}\}_{j=1}^z$ or their respective private resource $r_{p_i}$ for player $i$, i.e., $\mc{A}_i=\{r_{s_1},\ldots,r_{s_z},r_{p_i}\},~\forall i\in N\setminus \{n\}$.
The final player has exactly one action to use all the shared resources simultaneously $\mathcal{A}_n = (r_{s_1},\ldots,r_{s_z})$.
Each private resource $r_{p_i}$ has value $ v_{r_{p_i}} = f^{\rm g}(n)-\varepsilon$ with probability one where $\varepsilon >0$.
The value of the shared resources are random; each takes on value 1 w.p. $1/z < f^{\rm g}(n)-\varepsilon$ and 0 otherwise, and follow distribution $\mu_0$ such that exactly one is ever the high value.
When uninformed, each of the first $n-1$ players strictly prefers their private resource, giving a unique equilibrium welfare of $W(\Nash{a};\mu_0) = 1+(n-1)(f^{\rm g}(n)-\varepsilon)$.
Under the full information reveal policy $\Pi$, each of the $n-1$ agents strictly prefer to use the one shared resource of value 1, giving an expected welfare of $\ExpW(\BNash{\alpha};\mu_0,\Pi) = 1$.
Because each of the informed and uninformed equilibria are unique, this gives $\voip(\mu_0,\Pi) = \voim(\mu_0,\Pi) = \frac{1}{1+(n-1)(f^{\rm g}(n) - \varepsilon)} \rightarrow 1-\frac{1}{e}$ as $\varepsilon \rightarrow 0$ and $n \rightarrow \infty$.
\hfill\qed

Comparing the lower bounds of \cref{thm:mc} and \cref{prop:gair} highlights a trade-off between the guarantees of revealing information with the optimistic and pessimistic perspectives.
We further examine this trade-off in \cref{fig:VoI_tradeoff}.
Using \cref{thm:gen} and recent results on the price-of-anarchy/stability~\cite{Ramaswamy2022}, we can characterize lower bounds on $\voip$ and $\voim$ for different utility rules.

\section{Conclusion and Future Work}
We addressed the possible benefit and consequences of revealing information to local decision-makers in a distributed system.
By lower and upper bounding the value-of-informing, this work (1) quantified the possible effects of information revealing and (2) identified a trade-off between the guarantees of revealing information in the optimistic and pessimistic perspective.
Future work will answer the design question and develop methods to solve for information signaling policies that optimize expected system welfare.

\bibliographystyle{IEEEtran}
\bibliography{../../../../My_Library}

\end{document}